\documentclass[submission,copyright,creativecommons]{eptcs}
\providecommand{\event}{TARK 2019} 
\usepackage{settings}

\title{An Algebraic Approach for Action Based Default Reasoning}
\author{Pablo F. Castro
\institute{Universidad Nacional de R\'io Cuarto}
\institute{CONICET, Argentina}
\email{\tt pcastro@dc.exa.unrc.edu.ar}
\and
Valentin Cassano
\institute{Universidad Nacional de C\'ordoba}
\institute{CONICET, Argentina}
\email{vcassano@famaf.unc.edu.ar}
\and
Raul Fervari
\institute{Universidad Nacional de C\'ordoba}
\institute{CONICET, Argentina}
\email{fervari@famaf.unc.edu.ar}
\and
Carlos Areces
\institute{Universidad Nacional de C\'ordoba}
\institute{CONICET, Argentina}
\email{areces@famaf.unc.edu.ar}
}
\def\titlerunning{An Algebraic Approach for Action Based Default Reasoning}
\def\authorrunning{Castro, Cassano, Fervari \& Areces}
\begin{document}
\maketitle

\begin{abstract}
	Often, we assume that an action is permitted simply because it is not explicitly forbidden; or, similarly, that an action is forbidden simply because it is not explicitly permitted.
	This kind of assumptions appear, e.g., in autonomous computing systems where decisions must be taken in the presence of an incomplete set of norms regulating a particular scenario.
	Combining default and deontic reasoning over actions allows us to formally reason about such assumptions. 
	With this in mind, we propose a logical formalism for default reasoning over a deontic action logic.
	The novelty of our approach is twofold.
	First, our formalism for default reasoning deals with actions and action operators, and it is based on the deontic action logic originally proposed by Segerberg in~\cite{Segerberg1982}.
	Second, inspired by Segerberg's approach, we use tools coming from the theory of Boolean Algebra.
	These tools allow us to extend Segerberg's algebraic completeness result to the setting of Default Logics. 
\end{abstract}


\section{Introduction} 
\label{sec:intro}

The study of norms and their logical rules enjoys a renewed interest in the areas of Computer Science, Software Engineering, and Artificial Intelligence.
This renewed interest is  mainly due to the emergence of self-adaptive and autonomous systems such as self-driving cars, unmanned aerial vehicles, and flight-by-wire systems.
An important characteristic of these systems is that their behavior is regulated by a set of norms.
The following simple example illustrates this point.
A self-driving car cruising on a road ought to maintain a minimum speed and is forbidden to travel above a maximum speed.
Suppose now that the road is under construction; putting up traffic signs updating the minimum and maximum speeds requires the car to update its set of norms to account for the new road signs.
In addition, if some of the driving lanes are closed due to the construction work, some of the maneuvers that the car could take under normal road conditions may be restricted, e.g., passing a car is now forbidden. 
	
In this paper we propose a logical framework for reasoning about scenarios like the one described above. Our proposal is based on combining a \emph{deontic action logic} with \emph{default reasoning}.  
Deontic Logic, also called the logic of permission, prohibition, and obligation, has a rich history dating back to the pioneer work of von~Wright in \cite{vonWright:1951}.
At present, Deontic Logic encompasses a family of logical formalisms in which permission, prohibition, and obligation are captured as particular logical operators, commonly called deontic operators. 
In brief, deontic operators can roughly be categorized into so-called \emph{ought-to-be} and \emph{ought-to-do}~\cite{Aqvist:2002} depending on the kind of objects that they are applied to.
In the ought-to-be case, deontic logical operators are applied to propositions, e.g., \textit{the fence ought to be white}.
In ought-to-do case, deontic logical operators are applied to actions, e.g., \emph{killing is forbidden}.

Ought-to-do deontic operators prove to be more suitable when the focus of attention is on reasoning about actions and their deontological status.
Ought-to-do deontic operators typically have their origins in the work of Meyer in \cite{Meyer:1988} or in the work of Segerberg in \cite{Segerberg1982}.
In the first of these works, Meyer proposes to reduce deontic logical operators on actions to propositional dynamic logic and notions of violations. The resulting formalism is called
\emph{Dynamic Deontic Logic}.
In the second, Segerberg offers a novel definition of deontic logical operators on actions in set-theoretical terms using constructions coming from Boolean algebras. 
The resulting deontic logic is called \emph{Deontic Action Logic} (\DAL).
We favor \DAL because of its simplicity, and because its formalization of deontic operators relies on the theory of Boolean algebras, which enables us to obtain some desired results.

There are, however, some obstacles to the direct application of \DAL for reasoning in  scenarios such as the self-driving car discussed above.
First, some norms may need to be enacted in the presence of incomplete information --e.g., a self-driving car may tentatively assume that driving at a certain speed is permitted if no contrary information is known.
Second, some norms might be updated when more information is available, potentially overriding norms that are already in effect --e.g., in cases where the road is under construction.
To deal with these issues, we incorporate elements from \emph{default reasoning}.

Default reasoning is a non-monotonic formalism originally developed by Reiter in \cite{Reiter:1980}.
Default reasoning occupies a special place in its field due to its relatively simple syntax and semantics, its representation capabilities, and its relations to other non-monotonic formalisms (see, e.g.,~\cite{Antoniou:1997,Antoniou:2007}).
Our interest on default reasoning lays on the fact that it enables us to formally reason from tentative assumptions made in the absence of complete information, and that it can easily represent changing scenarios.
In this article, we show how default reasoning can be seamlessly integrated with \DAL.


\subsection{Contributions}
\label{subsec:contrib}

	First, we present a logical formalism to perform default reasoning over \DAL.
	As far as we are aware, our logic is the first of its kind. We consider that the formal tools provided by the combination of these two logics are useful for reasoning about self-adaptive and autonomous systems which have to fulfil different kinds of norms in partially described scenarios.
	Second, our formalization of default reasoning over \DAL is carried out by means of algebraic notions.
	This is in contrast to standard approaches to default reasoning, which are usually defined via meta-logical notions such as sets of consequences, maximal consistent sets, etc.
	Our formalization only relies on the algebraic semantics of the logic, and it is entirely constructed without resorting to meta-logical notions. A main benefit of this approach is that it enables us to extend the completeness result for \DAL presented in~\cite{Segerberg1982} to a completeness result for default reasoning on \DAL.

\subsection{Structure}
\label{subsec:struct}
In \Cref{sec:background} we introduce the basic notions needed to tackle the rest of the paper, including a brief introduction to the theory of Boolean algebras, \DAL, and default reasoning.
In \Cref{sec:theory} we introduce default deontic operators which allow us to perform default reasoning over \DAL.
This section contains our main contributions; therein, we prove some properties of our formalism, and show a completeness result.
In \Cref{sec:examples} we present a basic example with the aim of illustrating the notions introduced in earlier sections.
Finally, in \Cref{sec:final} we discuss some related work, and describe some of the further work that we plan to undertake.


%
%


\section{Preliminaries}
\label{sec:background}

In this section we introduce the basic notions needed to tackle the rest of the paper.

\subsection{Boolean Algebra in a Nutshell} 

We assume familiarity with the theory of Boolean algebras; and point out to \cite{Halmos:2009} for details.

\begin{definition}\label{definition:boolean:algebra}
	A \emph{Boolean algebra} is a structure
		$\mathbf{A} = \langle A, {+}, {\cdot}, {-}, 0, 1\rangle$
	where:
		$A$, also denoted $|\mathbf{A}|$, is non-empty set of elements called the \emph{carrier} set;
		$+$ and $\cdot$ are binary operators on $A$;
		$-$ is a unary operator on $A$;
		and $\{0,1\} \subseteq A$ are distinguished elements of $A$.
	We omit the axioms of Boolean algebras as they are well-known.
\end{definition}

\begin{definition}\label{definition:boolean:algebra:ideal}
	An \emph{ideal} of a Boolean algebra $\mathbf{A} = \langle A, {+}, {\cdot}, {-}, 0, 1\rangle$ is a subset $I \subseteq A$ s.t.:
	\begin{enumerate}[(i)]
		\setlength{\itemsep}{0pt} 
		\item if $x,y \in I$, ${x+y} \in I$
		and
		\item if $x \in I$, then ${x \cdot a} \in I$ for any $a \in A$.
	\end{enumerate}
	We define the set of all ideals of a $\mathbf{A}$ as $\Ideals[\mathbf{A}] = \set{I \subseteq A}{\text{$I$ is an ideal}}$.
	Moreover, for $B \subseteq A$, we define the \emph{ideal generated by $B$}, written $\GeneratedIdeal[B]$, as
		$\GeneratedIdeal[B] =
			{\bigcap\set{I \in \Ideals[\mathbf{A}]}{B \subseteq I}}$.
\end{definition}

\begin{definition}\label{definition:boolean:algebra:atom}
	Every Boolean algebra $\mathbf{A} = \langle A, {+}, {\cdot}, {-}, 0, 1\rangle$ is equipped with a relation $\sqsubseteq_{\mathbf{A}}$ defined as $a \sqsubseteq_{\mathbf{A}} b$ iff $a = {a \cdot b}$.
	The relation $\sqsubseteq_{\mathbf{A}}$ is a partial order.
	
\end{definition}

Henceforth, by a Boolean algebra, we mean a Boolean algebra which is not degenerate, i.e., which is such that $0 \neq 1$.
The concepts of Boolean algebras just introduced play a major role in what follows.


\subsection{Deontic Action Logic}\label{section:dal}

We cover the basis of a deontic action logic called {\DAL}.
{\DAL} is first introduced by Segerberg in \cite{Segerberg1982}.
We introduce its syntax and semantics following closely~\cite{Segerberg1982}, but including some remarks made in~\cite{Castro:2017,Trypuz15}.

\begin{definition}
The syntax of {\DAL} is comprised of \emph{actions} and \emph{formulas} defined on a countable set $\BasicActions = \set{a_i}{i\geq 0}$ of basic action symbols.
The set $\Actions$ of all actions of \DAL is given by the grammar:
\begin{nscenter}
$
\begin{array}{
r@{~::=~}l @{\quad~\quad} r@{~::=~}l}
\alpha &
		a_i
	\mid
		{\alpha \sqcup \alpha}
	\mid
		{\alpha \sqcap \alpha}
	\mid
		{\overline{\alpha}}
	\mid
		0
	\mid
		1.
\end{array}
$
\label{definition:actions}
\end{nscenter}
The set $\Formulas$ of all formulas of \textsf{DAL} is given by the grammar:
\begin{nscenter}
	$
\begin{array}{
r@{~::=~}l @{\quad~\quad} r@{~::=~}l}
\varphi &
		{\Not \varphi}
	\mid
		{\varphi \Or \varphi}
	\mid
		{\alpha=\beta}
	\mid
		\Permitted[\alpha]
	\mid
		\Forbidden[\alpha].
\end{array}
$
\label{definition:wffs}
\end{nscenter}
\end{definition}

We use $\alpha$, $\beta$, $\dots$, as variables for actions, and $A$, $B$, $\dots$, as variables for sets of actions.
Intuitively,
	any $a_i \in \BasicActions$ is a \emph{basic} action;
	$\alpha \sqcup \beta$ is the \emph{free-choice} between $\alpha$ and $\beta$;
	$\alpha \sqcap \beta$ is the \emph{parallel} execution of $\alpha$ and $\beta$;
	$\overline{\alpha}$ is the \emph{complement} of $\alpha$, i.e., any action other than $\alpha$;
	and
	$0$ and $1$ are the \emph{impossible} and the \emph{universal} actions, respectively.
We write $\alpha \equiv \beta$ for $(\alpha \sqcap \beta) \sqcup (\overline{\alpha} \sqcap \overline{\beta})$; and $\alpha \not \equiv \beta$ for $(\alpha \sqcap \overline{\beta}) \sqcup (\overline{\alpha} \sqcap \beta)$. 

%
We use $\varphi$, $\psi$, $\dots$, and $\Phi$, $\Psi$, $\dots$, as variables for formulas, and sets of formulas, respectively.
The Boolean logical connectives $\Not$ and $\Or$ have their usual intuitive understanding: $\Not$ stands for \emph{negation}; and $\Or$ stands for \emph{disjunction}.
We also consider derived Boolean logical connectives: $\True$ for \emph{verum}, $\False$ for \emph{falsum}, $\And$ for \emph{conjunction}, and $\Implies$ for \emph{material implication}.
These derived Boolean connectives are defined from $\Not$ and $\Or$ in the usual way.
The formula $\alpha = \beta$ intuitively means that $\alpha$ and $\beta$ are \emph{equal}.
The deontic operator $\Permitted$ stands for \emph{permitted} and it intuitively means that $\alpha$ is allowed to be executed.
In turn, the deontic operator $\Forbidden$ stands for \emph{forbidden} and it intuitively means that the execution of $\alpha$ forbidden.

The semantics of \DAL is given by \emph{deontic action algebras} (which contain Boolean algebras and ideals as an integral part) and \emph{valuation} functions.
These concepts are made precise below.

\begin{definition}\label{definition:deontic:action:algebra}
	A \emph{deontic action algebra} is a triple
	$\DeonticStructure = \langle
		\EventsAlgebra,
		\PermittedEvents,
		\ForbiddenEvents
	\rangle$ in which:
		$\mathbf{E} =
			\langle
				E, {+}, {\cdot}, {-}, 0, 1
			\rangle$
		is a Boolean algebra, and $\PermittedEvents$ and $F$ are ideals of $\EventsAlgebra$ s.t.\ ${\PermittedEvents \cap \ForbiddenEvents} = \{0\}$.
\end{definition}

Intuitively, given a deontic action algebra $\DeonticStructure = \langle \EventsAlgebra, \PermittedEvents, \ForbiddenEvents \rangle$, we can think of $\EventsAlgebra$ as an algebra of events, i.e., possible outcomes of actions, and of $\PermittedEvents$ and $\ForbiddenEvents$ as sets of permitted and forbidden events.
The condition ${\PermittedEvents \cap \ForbiddenEvents} = \{0\}$ for $\DeonticStructure$ can be understood as: \emph{only an impossible action is both permitted and forbidden}.
As pointed out in \cite{Trypuz15}, Boolean algebras are just one of the possible structures that can express a space of attitudes towards actions; yet, they are simple structures with a significant expressive power.

\begin{definition}\label{definition:valuation:function}
	Let $\DeonticStructure = \langle \EventsAlgebra, \PermittedEvents, \ForbiddenEvents \rangle$ be a deontic action algebra; a valuation function for $\DeonticStructure$ is a function $\InterpretationFunction : {\BasicActions \rightarrow |\EventsAlgebra|}$ which maps basic actions to events.
	$\InterpretationFunction$ extends uniquely to the set $\mathsf{Act}$ of actions as:
	\begin{nscenter}
		$
		\begin{array}{r@{\;=\;}l}
		\InterpretationFunction(0)
			&  0 \\
		\InterpretationFunction(1)
			&  1 \\
		\InterpretationFunction(\alpha \sqcup \beta)
			&
			{\InterpretationFunction(\alpha)
			+
			\InterpretationFunction(\beta)} \\
		\InterpretationFunction(\alpha \sqcap \beta)
			&
			{\InterpretationFunction(\alpha)
			\cdot
			\InterpretationFunction(\beta)} \\
		\InterpretationFunction(\overline{\alpha})
			&
			{-\InterpretationFunction(\alpha)}.
		\end{array}
		$
	\end{nscenter}
\end{definition}

%

\begin{definition}\label{definition:dal:consequence}
	Let $\DeonticStructure = \langle \EventsAlgebra, \PermittedEvents, \ForbiddenEvents \rangle$ be a deontic action algebra, and $\InterpretationFunction$ be a valuation for $\DeonticStructure$; the notion of a formula $\varphi$ being satisfied in $\DeonticStructure$ under $\InterpretationFunction$, notation ${\DeonticStructure, \InterpretationFunction} \vDash \varphi$, is inductively defined as:
	\begin{nscenter}
	$
	\begin{array}{lcl}
		{\DeonticStructure, \InterpretationFunction} \vDash {\Not\varphi}
			& \text{iff} &
			{\DeonticStructure, \InterpretationFunction} \not\vDash \varphi \\
		{\DeonticStructure, \InterpretationFunction} \vDash {\varphi \Or \psi}
			& \text{iff} &
			{\DeonticStructure, \InterpretationFunction} \vDash \varphi
			\mbox{ or }
			{\DeonticStructure, \InterpretationFunction} \vDash \psi\\
		{\DeonticStructure, \InterpretationFunction} \vDash {\alpha=\beta}
			& \text{iff} &
			\InterpretationFunction(\alpha) = \InterpretationFunction(\beta) \\
		{\DeonticStructure, \InterpretationFunction} \vDash {\Permitted}
			& \text{iff} &
			\InterpretationFunction(\alpha) \in \PermittedEvents \\
		{\DeonticStructure, \InterpretationFunction} \vDash {\Forbidden}
			& \text{iff} &
			\InterpretationFunction(\alpha) \in \ForbiddenEvents.
	\end{array}
	$
\end{nscenter}
We say that a formula $\varphi$ is an algebraic consequence of a set of formulas $\Phi$, notation $\Phi\consequence\varphi$, iff for any deontic action algebra $\DeonticStructure$ and valuation $\InterpretationFunction$ for $\DeonticStructure$, if ${\DeonticStructure, \InterpretationFunction} \vDash \psi$ for all $\psi \in \Phi$, then ${{\DeonticStructure, \InterpretationFunction} \vDash \varphi}$.
\end{definition}

Thus far we have treated \DAL from an algebraic perspective.
We now turn our attention to an axiom system and a Hilbert-style notion of provability for \DAL.

\begin{definition}\label{definition:dal:provability}
The standard list of axioms for \DAL consists of:
\begin{enumerate}
	\setlength{\itemsep}{0pt} 
	\item a complete (classical) set of axioms for $\Not$, and $\Or$ (together with $\True$, $\False$, and $\Implies$);
	\item a complete set of Boolean algebra axioms for $\sqcup$, $\sqcap$, $\overline{\phantom{\alpha}}$, $0$ and $1$; together with the axiom $\Not (0 = 1)$;
	\item a complete set of axioms for equality for $=$;
	\item the substitution axiom ${\alpha=\beta} \Implies (\varphi \Implies {\varphi_{\alpha}^{\beta}})$, where $\varphi_{\alpha}^{\beta}$ is the formula obtained from replacing some ocurrences of $\alpha$ with $\beta$;
	\item the deontic axioms
	\begin{enumerate}[D1.]
		\item $\Permitted[\alpha\sqcup\beta] \Iff (\Permitted[\alpha] \And \Permitted[\beta])$;
		\item $\Forbidden[\alpha\sqcup\beta] \Iff (\Forbidden[\alpha] \And \Forbidden[\beta])$;
		\item ${\alpha = 0} \Iff (\Permitted[\alpha] \And \Forbidden[\alpha])$.
	\end{enumerate}
\end{enumerate}
Let ${\Phi \cup \varphi}$ be a set of formulas;%
\footnote{We sometimes use the notation ${A \cup a}$ instead of $A \cup\{a\}$.}
consider a finite sequence $s = \psi_1, \dots, \psi_n$ of formulas s.t.\ $\psi_n = \varphi$ and for each $k \leq n$, $\psi_k$ is either:
\begin{enumerate}[(i)]
	\setlength{\itemsep}{0pt} 
	\item an axiom of \DAL;
	\item a member of $\Phi$;
	\item obtained from two earlier formulas in $s$ by \emph{modus ponens}, i.e., there are ${i,j} < k$ s.t.\ $\psi_j = {\psi_i \Implies \psi_k}$.
\end{enumerate}
We call any such a sequence $s$ a proof of $\varphi$ from $\Phi$.
We say that $\varphi$ is provable from $\Phi$, written $\Phi \entails \varphi$, if there is a proof of $\varphi$ from $\Phi$.
We define $\ClosureDAL{\Phi} = \set{\varphi \in \Formulas}{\Phi \entails \varphi}$.
We say that $\Phi$ is $\entails$-consistent iff $\ClosureDAL{\Phi} \subsetneq \Formulas$ (alternatively, iff $\Phi \not\entails \False$).
\end{definition}

\begin{proposition}\label{proposition:dal:soundness:completeness}
$\Phi \consequence \varphi$ iff $\Phi \entails \varphi$.
\end{proposition}

\Cref{proposition:dal:soundness:completeness} is proven by Segerberg in \cite{Segerberg1982}, and it establishes that the proof system from \Cref{definition:dal:provability} is \emph{strongly complete} with respect to the semantics based on deontic action algebras.
The crucial step carried out by Segerberg in the proof of \Cref{proposition:dal:soundness:completeness} is the construction of a Lindenbaum-Tarski algebra, and a pair of ideals in this algebra, which serves as a canonical deontic action algebra for establishing completeness.  
We present this construction in detail in \Cref{sec:extensions:algebraic} and use it to show how it can be extended to obtain the main result of this paper.


\subsection{Propositional Default Logic} \label{sec:default:logic}

We present a brief outline of Default Logic~\cite{Reiter:1980}.
Our aim is to recall some basic definitions to make our work self-contained.
In particular, we wish to bring to the fore a simple definition of default consequence following Makinson in \cite{Makinson:2005}.
To simplify our exposition, we restrict our definitions to Classical Propositional Logic (\CPL)~\cite{vanDalen:2004}.
This means that, in this section,
	by a formula we will mean a formula of \CPL.
We also use $\entails^{\CPL}$ to indicate the provability relation of \CPL and $\Closure{\Phi}$ for the set $\set{\varphi}{\Phi \entails^{\CPL} \varphi}$.


We take as our starting point the concept of a \emph{default} as an expression
$\idefaultrule$, 
where $\prerequisite$, $\justification$, and $\consequent$ are formulas called \emph{prerequisite}, \emph{justification}, and \emph{consequent}, respectively.
We use $\Delta$ as a variable for a set of defaults.
Intuitively, we can think of a default $\idefaultrule$ as a rule enabling us to pass from $\prerequisite$ to $\consequent$, provided that we can establish $\prerequisite$ and that the construction that we use for establishing $\prerequisite$ is individually consistent with $\Justifications\cup\justification$; where $\Justifications$ is the set of justifications of the defaults used in the aforementioned construction.
This intricate notion is formalized in \Cref{definition:extension:reiter}.

\begin{definition}[\cite{Reiter:1980}]\label{definition:extension:reiter}
	Let $\Phi$ be a set of formulas and let $\Delta$ be a set of defaults;
	also, let $\rmGamma^{\Phi}_{\Delta}$ be a function s.t.\
		for all sets of formulas $\Psi$,
			$\rmGamma^{\Phi}_\Delta(\Psi)$ is the smallest set of formulas which satifies:
	\begin{enumerate}[(i)]
		\setlength{\itemsep}{0pt} 
		\item
			$\Phi \subseteq \rmGamma^{\Phi}_{\Delta}(\Psi)$
		\item
			$\rmGamma^{\Phi}_{\Delta}(\Psi) = \Closure{(\rmGamma^{\Phi}_{\Delta}(\Psi))}$
		\item
			For all ${\idefaultrule \in \Delta}$, 
			if
				$\prerequisite \in \rmGamma^{\Phi}_{\Delta}(\Psi)$
				and
				${\Not \justification} \notin \Psi$,
			then,
				$\consequent \in \rmGamma^{\Phi}_{\Delta}(\Psi)$.
	\end{enumerate}
	We say that $\rmEpsilon$ is an \emph{extension} of $\Phi$ under $\Delta$ iff it is a fixed point of $\rmGamma_{\Delta}^{\Phi}$, i.e., iff $\rmEpsilon = \rmGamma_{\Delta}^{\Phi}(\rmEpsilon)$.
	We define the set of all extensions of $\Phi$ under $\Delta$ as $\Extensions_{\Delta}^{\Phi} = \set{\rmEpsilon}{{\rmEpsilon = \rmGamma_{\Delta}^{\Phi}(\rmEpsilon)}}$.
\end{definition}

\begin{example}\label{ex:nunique}
	Let $\Phi = \{ p, {\Not q \Or \Not r} \}$ and $\Delta = \{ {\idefaultrule[p][q][q]}, {\idefaultrule[p][r][r]} \}$;
	it follows that 
	$\rmEpsilon_1 = \Closure{\{p,{\Not q},r\}}$ and $\rmEpsilon_2 = \Closure{\{p,{\Not r},q\}}$ are extensions of $\rmGamma_{\Delta}^{\Phi}$.
\end{example}

Extensions as in \Cref{definition:extension:reiter} can be viewed as sets of formulas which are closed under the application of defaults.
This yields a notion of default consequence in the following sense.

\begin{definition}\label{definition:default:consequence:basic}
	Let ${\Phi \cup \varphi}$ be a set of formulas and $\Delta$ be a set of defaults; we say that $\varphi$ is a \emph{default consequence} of $\Phi$ under $\Delta$, written $\Phi \econsequence^{\CPL}_{\Delta} \varphi$, iff $\rmEpsilon \entails^{\CPL} \varphi$ for some $\rmEpsilon \in \Extensions^{\Phi}_{\Delta}$.
\end{definition}

The relation $\econsequence^{\CPL}_{\Delta}$ in \Cref{definition:default:consequence:basic} is called \emph{credulous} in the literature on Default Logic.
For this relation, it can be proven that the \emph{principle of monotonicity} does not necessarily hold, i.e., it is not necessarily the case that if $\Phi \econsequence^{\CPL}_{\Delta} \varphi$, then ${\Phi \cup \Psi} \econsequence^{\CPL}_{\Delta} \varphi$.
Whether or not monotonicity holds for $\econsequence^{\CPL}_{\Delta}$ depends on the particular set $\Delta$ of defaults.
We take failure of monotonicity for $\econsequence^{\CPL}_{\Delta}$ as a desirable property for its modelling capabilities of real world phenomena.

\begin{definition}\label{definition:default:consequence:interpretation}
	We say that $\econsequence^{\CPL}_{\Delta}$ \emph{interprets} $\entails^{\CPL}$ iff if $\Phi \entails^{\CPL} \varphi$, then $\Phi \econsequence^{\CPL}_{\Delta} \varphi$.
\end{definition}

\begin{property}\label{proposition:default:consequence:interpretation}
	$\econsequence^{\CPL}_{\Delta}$ interprets $\entails^{\CPL}$ iff for all sets of formulas $\Phi$, $\Extensions^{\Phi}_{\Delta} \neq \emptyset$.
\end{property}

\Cref{definition:default:consequence:interpretation} imposes a basic condition of $\econsequence^{\CPL}_{\Delta}$ which we also take as desirable.
We view default consequence as an enlargement of an underlying notion of provability.
%
%
%
It is worth noticing that `interpretability' depends on the existence of extensions.
Unfortunately, Reiter shows in \cite{Reiter:1980} that this may fail for some sets $\Delta$ of defaults.
Thus, `interpretability' is not guaranteed for arbitrary $\econsequence^{\CPL}_{\Delta}$.
This hinders our treatment of $\econsequence^{\CPL}_{\Delta}$.
At this point, we can go down two possible paths: (i) modify \Cref{definition:extension:reiter} to guarantee the existence of extensions; (ii) single out defaults for which extensions are guaranteed to exist.
As to (i), among the most popular modifications of \Cref{definition:extension:reiter} which guarantee the existence of extensions we have: \emph{justified} extensions, proposed by {\L}ukaszewicz in~\cite{Lukaszewicz:1988}; and \emph{constrained} extensions, proposed by Delgrande~et~al.~in~\cite{Delgrande:1994}.
As to (ii), we have the set of \emph{normal} defaults as a very large and natural set of defaults for which extensions as in \Cref{definition:extension:reiter} are guaranteed to exist~\cite{Reiter:1980}.
We choose to go down the second path and to restrict our attention to the case of normal defaults.
We make this restriction precise in \Cref{definition:default:consequence:normal}.

\begin{definition}\label{definition:default:consequence:normal}
	We say that a default $\idefaultrule$ is \emph{normal} iff $\justification = \consequent$.
	We use $\normal$ as notation for a normal default.
	A set $\Delta$ of defaults is normal iff all defaults in $\Delta$ are normal.
	We say that $\econsequence^{\CPL}_{\Delta}$ is \emph{normal} iff $\Delta$ is normal.
\end{definition}

\begin{property}\label{proposition:default:consequence:interpretation:normal}
	If $\econsequence^{\CPL}_{\Delta}$ is \emph{normal}, then $\econsequence^{\CPL}_{\Delta}$ interprets $\entails^{\CPL}$.
\end{property}

\begin{property}\label{proposition:default:consequence:normal:consistency}
	If $\econsequence^{\CPL}_{\Delta}$ is \emph{normal}, then $\Phi \econsequence^{\CPL}_{\Delta} \False$ iff $\Phi \entails^{\CPL} \False$.
\end{property}

Intuitively, we can understand \Cref{proposition:default:consequence:normal:consistency} as stating that defaults cannot be a source of inconsistency.
As a final remark, it is a known result that extensions, justified extensions, and constrained extensions, coincide for normal defaults~\cite{Froidevaux:1994,Cassano:2019}.
Normal defaults also arise often in many application areas.
Thus, restricting ourselves to normal defaults is not too confining.
Furthermore, normal default consequence as in \Cref{definition:default:consequence:normal} does not guarantee monotonicity, i.e., there are normal sets of defaults $\Delta$ for which $\Phi \econsequence^{\CPL}_{\Delta} \varphi$ and ${\Phi \cup \Psi} \not \econsequence^{\CPL}_{\Delta} \varphi$.

\section{Default Deontic Action Logic}
\label{sec:theory}

In this section we present the main results of our work.
We begin by introducing a definition of normal default consequence for \DAL.
This notion of default consequence enables us to perform default reasoning over deontic operators applied to actions.
For this notion, we develop a Hilbert-style proof calculus with a consistency check which enables us to capture default reasoning steps.
Moreover, we show a completeness result for our calculus extending the method proposed by Segerberg in \cite{Segerberg1982}.

\subsection{Normal Default Consequence on \DAL} 

We bring attention to the fact that the relation $\econsequence^{\CPL}_{\Delta}$ of normal default consequence presented in \Cref{definition:default:consequence:normal} is parametric on \CPL.
In other words, it is possible to define a notion of normal default consequence $\econsequence^{\DAL}_{\Delta}$ for \DAL simply by replacing $\entails^{\CPL}$ for $\entails$. 
It follows directly from this definition that $\econsequence^{\DAL}_{\Delta}$ is non-monotonic, i.e., monotonicity fails for $\Delta$ an arbitrary set of normal defaults in \DAL, and that $\econsequence^{\DAL}_{\Delta}$ interprets $\entails$.
To simplify  notation, from now on we write $\econsequence_{\Delta}$ instead of $\econsequence^{\DAL}_{\Delta}$.



\subsection{Proofs for Normal Default Consequence on \DAL} 

We assume that $\Delta$ is an arbitrary but fixed set of normal defaults defined on $\Formulas$ and that $\econsequence_{\Delta}$ is the normal default consequence relation associated to $\Delta$.
We present a Hilbert-style notion of proof for $\econsequence_{\Delta}$.

\begin{definition}\label{definition:default:provability}
Let $\Phi$ be a set of formulas and $\Delta$ be a set of defaults; also let $s = {\psi_1, \dots, \psi_n}$ be a finite sequence of formulas s.t.\ $\psi_n = \varphi$ and for each $k \leq n$, $\psi_k$ is either:
\begin{enumerate}[(i)]
	\setlength{\itemsep}{0pt} 
	\item an axiom of \DAL;
	\item a member of $\Phi$;
	\item obtained from two earlier formulas in $s$ by \emph{modus ponens}, i.e., there are ${i,j} < k$ s.t.\ $\psi_j = {\psi_i \Implies \psi_k}$;
	\item obtained from an earlier formula in $s$ by \emph{default detachment}, i.e., there is $j < k$ s.t.\ ${\normal[\psi_j][\psi_k]} \in \Delta$.
\end{enumerate}
If such a sequence $s$ exists, and $\set{\psi_i}{ 1 \leq i \leq n}$ is $\entails$-consistent, we say that $s$ is a default proof of $\varphi$ from $\Phi$ under $\Delta$.
Moreover, we say that $\varphi$ is $\Delta$-provable from $\Phi$, and write $\Phi \dentails_{\Delta} \varphi$, if there is a default proof of $\varphi$ from $\Phi$ under $\Delta$.
\end{definition}

The notion of a default proof in \Cref{definition:default:provability} can also be formulated inductively.
In this inductive formulation, each application of default detachment needs of a consistency check w.r.t.\ the formulas already in the proof; i.e., if default detachment is to be applied in a step $k$ in the proof, then, it is required for the set $\set{\psi_i}{ 1 \leq i \leq k}$ to be $\entails$-consistent.
This inductive formulation is equivalent to \Cref{definition:default:provability}.

\begin{theorem}\label{proposition:default:proof:extensions}
	For any $\entails$-consistent set $\Phi$ of formulas of $\Formulas$; $\Phi \dentails_{\Delta} \varphi$ iff $\Phi \econsequence_{\Delta} \varphi$.
\end{theorem}

\begin{proof}
	We use an alternative characterization of extensions in terms of \emph{closed generating sequences} (which adapts a definition of a closed process presented by Antoniou in \cite{Antoniou:1997}).
	By a $\Delta$-sequence we mean a (potentially infinite) sequence $s = {s_1, s_2, s_3, \dots}$ of defaults of $\Delta$.
	Let $s = {s_1, s_2, s_3, \dots}$ be a $\Delta$-sequence; the following notation is useful:
		$s|_{n} = {s_1, \dots, s_n}$,
		${s_i = {\normal[\prerequisite_i][\consequent_i]}}$,
		and
		$\Consequents_s = \set{\consequent_i}{s_i = {\normal[\prerequisite_i][\consequent_i]}}$.
	A \emph{generating sequence} is a $\Delta$-sequence $s = {s_1, s_2, s_3, \dots}$ s.t.\  for all indices $i$ of $s$, (a) ${\Phi \cup \Consequents_{s|_{(i-1)}}} \entails \prerequisite_i$ and (b) ${\Phi \cup \Consequents_{s|_i}}$ is $\entails$-consistent.
	A generating sequence is \emph{closed} iff it is not a strict initial segment of any other generating sequence.
	It can be proven, by adapting the proof found in \cite{Antoniou:1997}, that every closed generating sequence yields an extension, and that every extension has an associated closed generating sequence.
	
	Turning to the proof of \Cref{proposition:default:proof:extensions}, we first prove that if $\Phi \econsequence_{\Delta} \varphi$, then $\Phi \dentails_{\Delta} \varphi$.
	Let $E$ be an extension; then $E = \ClosureDAL{(\Phi \cup \Consequents_s)}$ where $s = {s_1, s_2, s_3, \dots}$ is a generating sequence.
	If ${\Phi \cup \Consequents_s} \entails \varphi$, from monotonicity and compactness for $\entails$, we obtain that for some index $n$ of $s$, ${\Phi \cup \Consequents_{s|_n}} \entails \varphi$.
	Let $p$ be a proof of $\varphi$ from ${\Phi \cup \Consequents_{s|_n}}$;
	we extend $p$ by:
		(1) inserting in front of the first occurrence of some $\consequent_m \in  (\Consequents_{s|_n} \cap p)$ a proof of $\prerequisite_m$ from $\Phi \cup \Consequents_{s|_(m-1)}$ (marking all successive occurrences of $\consequent_m$ as treated);
		(2) repeating (1) until all $\consequent_m \in  ({\Consequents_{s|_n} \cap p})$ have been treated.
	The result is a finite sequence of formulas which is, by construction, a default proof of $\varphi$ from $\Phi$ under $\Delta$.
	Thus, if ${\Phi \econsequence_{\Delta} \varphi}$, ${\Phi \dentails_{\Delta} \varphi}$.
	
	We now prove that if $\Phi \dentails_{\Delta} \varphi$, then $\Phi \econsequence_{\Delta} \varphi$.
	Let $p$ be a default proof of $\varphi$ from $\Phi$ under $\Delta$ and $s$ be the sequence of defaults of $\Delta$ used in $p$ in their order of appearance; by construction $s$ is a generating sequence.
	Extending $s$ to a generating sequence $s'$ that is closed we obtain that $\Phi \econsequence_{\Delta} \varphi$.
	Thus, if $\Phi \dentails_{\Delta} \varphi$, then $\Phi \econsequence_{\Delta} \varphi$.
\end{proof}

\subsection{Algebraic Extensions of Basic Deontic Defaults}
\label{sec:extensions:algebraic}

By a \emph{basic deontic default} we mean a normal default $\normal$ s.t.\ $\prerequisite = \Permitted[\alpha]$ and $\consequent = \Permitted[\beta]$; or $\prerequisite = \Forbidden[\alpha]$ and $\consequent = \Forbidden[\beta]$.
We write basic deontic defaults as $\Permitted[{\normal[\alpha][\beta]}]$ or $\Forbidden[{\normal[\alpha][\beta]}]$. 
Basic deontic defaults gain in interest when they are thought of as capturing defeasible conditional notions of permission and prohibition on actions.
We elaborate on the formal machinery behind basic deontic defaults by algebraic means using a Lindenbaum-Tarski construction.
In the context of deontic action logics, this construction is originally proposed in \cite{Segerberg1982} to show completeness of \DAL. 
The fundamental result of this section is the extension of Segerberg's result to account for default provability and consequence defined on basic deontic defaults.
For the rest of this section we assume that $\Phi$ is an arbitrary but fixed $\entails$-consistent set of formulas, and that $\Delta$ is an arbitrary but fixed set of basic deontic defaults.

We begin with a standard algebraic construction.
Define a binary relation $\equiv_{\Phi}$ on $\Actions$ as:
	\begin{nscenter}
	$
	{\alpha \equiv_{\Phi} \beta}
	\qquad \text{iff} \qquad
	{{\alpha = \beta} \in \ClosureDAL{\Phi}}.
	$
	\end{nscenter}
The relation $\equiv_{\Phi}$ is an equivalence relation.
Thus, for any action $\alpha \in \Actions$, we can define the equivalence class of $\alpha$ under $\equiv_{\Phi}$ as
	\begin{nscenter}
	$	[\alpha]_{\Phi}
			= \set{ \beta }{ \alpha \equiv_{\Phi} \beta}.
	$
	\end{nscenter}
We use $\Actions / \Phi$ to denote the collection of all equivalence classes of $\equiv_{\Phi}$, i.e., for the quotient of $\Actions$ under $\equiv_{\Phi}$.
Next, we define the Lindenbaum-Tarski algebra for $\Phi$ as the structure: 
	\begin{nscenter}
	$
		\LTAlgebra_{\Phi}=\langle {\Actions / \Phi},
		{+}, {\cdot}, {-}, 0, 1 \rangle
	$,
	\end{nscenter}
where the operations $+$, $\cdot$, $-$, and the distinguished elements $0$ and $1$ are defined as:
	\begin{align*}
		[\alpha]_\Phi + [\beta]_\Phi &= [\alpha \sqcup \beta ]_\Phi \\
		[\alpha]_\Phi \cdot [\beta]_\Phi &= [\alpha \sqcap \beta ]_\Phi \\
		{-[\alpha]_\Phi} &= [\overline{\alpha}]_\Phi \\
		0 &= [0]_\Phi \\
		1 &= [1]_\Phi.
	\end{align*}
It is trivial to prove that the operations on $\LTAlgebra_{\Phi}$ are well-defined.
%
The notions of a permitted ideal $\PermittedEvents_{\LTAlgebra_{\Phi}}$ and a forbidden ideal $\ForbiddenEvents_{\LTAlgebra_{\Phi}}$ for $\LTAlgebra_{\Phi}$ are defined as:
	\begin{nscenter}	
	$
	\begin{array}{l@{\;=\;}l}		
		\PermittedEvents_{\LTAlgebra_{\Phi}}
			&
			\textstyle \bigcap\set{ P \in \Ideals[\LTAlgebra_{\Phi}]}
				{
				\text{
				if
					$\Permitted[\alpha] \in \ClosureDAL{\Phi}$,
				then
					$[\alpha]_{\Phi} \in P$}
				}
		\\
		\ForbiddenEvents_{\LTAlgebra_{\Phi}}
			&
			\textstyle \bigcap\set{ F \in \Ideals[\LTAlgebra_{\Phi}]}
				{
				\text{
				if
					$\Forbidden[\alpha] \in \ClosureDAL{\Phi}$,
				then
					$[\alpha]_{\Phi} \in F$}
				}.
	\end{array}
	$
	\end{nscenter}
The main result proven in \cite{Segerberg1982} is that the triple $\langle \LTAlgebra_{\Phi}, \PermittedEvents_{\LTAlgebra_{\Phi}}, \ForbiddenEvents_{\LTAlgebra_{\Phi}} \rangle$ is a deontic action algebra.
This result, together with a function $\InterpretationFunction_{\LTAlgebra_{\Phi}}(a) = [a]_\Phi$ for all actions $a \in \Actions$, is then used to obtain a completeness result for provability $\entails$ and consequence $\consequence$ in \DAL; see \cite{Segerberg1982} for details.

We extend the construction of $\langle \LTAlgebra_{\Phi}, \PermittedEvents_{\LTAlgebra_{\Phi}}, \ForbiddenEvents_{\LTAlgebra_{\Phi}} \rangle$ from \cite{Segerberg1982}, in order to deal with basic deontic defaults in \Cref{definition:extension:algebraic}.
We begin with a preliminary definition.

\begin{definition}\label{definition:deontic:dual}
	Given $\PermittedEvents_{\LTAlgebra_{\Phi}}$ and $\ForbiddenEvents_{\LTAlgebra_{\Phi}}$; define:
	\begin{nscenter}	
		$
		\begin{array}{l@{\;=\;}l}
		\DeonticDual{\PermittedEvents_{\LTAlgebra_{\Phi}}} &
			\set
				{ [\beta]_\Phi }
				{\text{
					there is $[\alpha]_{\Phi} \in \LTAlgebra_{\Phi}$ s.t.\ ${{\Not \Permitted[\alpha]} \in \ClosureDAL{\Phi}}$ and ${[\beta]_\Phi \sqsubseteq_{\LTAlgebra_{\Phi}} [\alpha]_\Phi}$
				}
				} \setminus \PermittedEvents_{\LTAlgebra_{\Phi}}
		\\
		\DeonticDual{\ForbiddenEvents_{\LTAlgebra_{\Phi}}} &
			\set
			{ [\beta]_\Phi }
			{\text{
				there is $[\alpha]_{\Phi} \in \LTAlgebra_{\Phi}$ s.t.\ ${{\Not \Forbidden[\alpha]} \in \ClosureDAL{\Phi}}$ and ${[\beta]_\Phi \sqsubseteq_{\LTAlgebra_{\Phi}} [\alpha]_\Phi}$
			}
			} \setminus \ForbiddenEvents_{\LTAlgebra_{\Phi}}.
	\end{array}
	$
	\end{nscenter}
	We say that $\DeonticDual{\PermittedEvents_{\LTAlgebra_{\Phi}}}$ and $\DeonticDual{\ForbiddenEvents_{\LTAlgebra_{\Phi}}}$ are the \emph{deontic duals} of $\PermittedEvents_{\LTAlgebra_{\Phi}}$ and $\ForbiddenEvents_{\LTAlgebra_{\Phi}}$, resp.
	We define
	\begin{nscenter}
		$
		[\alpha]_\Phi \preccurlyeq \DeonticDual{\PermittedEvents_{\LTAlgebra_{\Phi}}}
		\text{ iff }
		{{\set{[\beta]_\Phi}{ [\beta]_\Phi \sqsubseteq_{\LTAlgebra_{\Phi}} [\alpha]_\Phi } \cap \DeonticDual{\PermittedEvents_{\LTAlgebra_{\Phi}}}} \neq \emptyset}.
		$
	\end{nscenter}
	The expression $[\alpha]_\Phi \preccurlyeq \DeonticDual{\ForbiddenEvents_{\LTAlgebra_{\Phi}}}$ is defined in a similar way.
\end{definition}

Let us note that $\DeonticDual{\PermittedEvents_{\LTAlgebra_{\Phi}}}$ is not necessarily a subset of $\ForbiddenEvents_{\LTAlgebra_{\Phi}}$, nor $\ForbiddenEvents_{\LTAlgebra_{\Phi}}$ is necessarily a subset of $\DeonticDual{\PermittedEvents_{\LTAlgebra_{\Phi}}}$; and similarly for $\DeonticDual{\ForbiddenEvents_{\LTAlgebra_{\Phi}}}$ and $\PermittedEvents_{\LTAlgebra_{\Phi}}$.
Deontic duals play a part in the check for consistency of basic deontic defaults.

\begin{definition}\label{definition:extension:algebraic}
	Let $\AlgebraicExtension_{\Phi}^{\Delta}: {{\Ideals[\LTAlgebra_{\Phi}]}^2 \rightarrow {\Ideals[\LTAlgebra_{\Phi}]}^2}$
		be a function s.t.\ if ${\AlgebraicExtension_{\Phi}^{\Delta}(\PermittedEvents, \ForbiddenEvents)} = (\PermittedEvents', \ForbiddenEvents')$, then $\PermittedEvents'$ and $\ForbiddenEvents'$ are the smallest ideals which satisfy:
	\begin{enumerate}[(i)]
		\itemsep 0cm
		\item $\PermittedEvents_{\LTAlgebra_{\Phi}} \subseteq \PermittedEvents'$ and $\ForbiddenEvents_{\LTAlgebra_{\Phi}} \subseteq \ForbiddenEvents'$;
		\item for all ${\Permitted[{\normal[\alpha][\beta]}]} \in \Delta$;
			if
				$[\alpha]_\Phi \in \PermittedEvents'$,
				${\GeneratedIdeal[\PermittedEvents' \cup {[\beta]_\Phi}] \cap F'} = [0]_\Phi$,
				and
				$[\beta]_\Phi \not\preccurlyeq \DeonticDual{\PermittedEvents_{\LTAlgebra_{\Phi}}}$,
			then,
				$[\beta]_\Phi \in \PermittedEvents'$; 
		\item for all ${\Forbidden[{\normal[\alpha][\beta]}]} \in \Delta$;
			if 
				$[\alpha]_\Phi \in \ForbiddenEvents'$,
				${\GeneratedIdeal[\ForbiddenEvents' \cup {[\beta]_\Phi}] \cap P'} = [0]_\Phi$,
				and
				$[\beta]_\Phi \not\preccurlyeq \DeonticDual{\ForbiddenEvents_{\LTAlgebra_{\Phi}}}$,
			then,
				$[\beta]_\Phi \in \ForbiddenEvents'$.
	\end{enumerate}
	$(P, F)$ is an \emph{algebraic extension} of $\Phi$ under $\Delta$ iff it is a fixed point of $\AlgebraicExtension_{\Phi}^{\Delta}$, i.e., iff $(P, F) = \AlgebraicExtension_{\Phi}^{\Delta}(P, F)$.
\end{definition}
	
In \Cref{definition:extension:algebraic} an algebraic extension is a pair of ideals in the Lindenbaum-Tarski $\LTAlgebra_{\Phi}$ enlarging the ideals $\PermittedEvents_{\LTAlgebra_{\Phi}}$ and $\ForbiddenEvents_{\LTAlgebra_{\Phi}}$ in a consistent way.
This construction is depicted in \Cref{figure:algebraic:extension}.

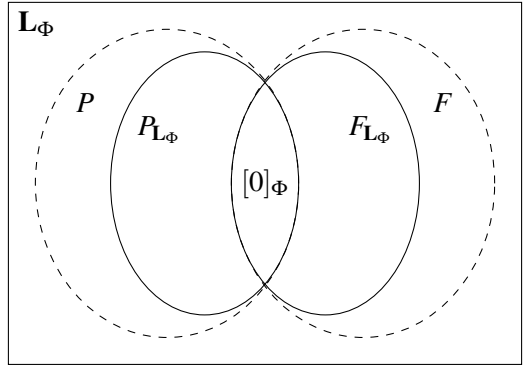
\begin{figure}[!h]
\begin{center}
\begin{tikzpicture}
	[node distance=0cm]
	\node
	[]
		(Zero)
		{$[0]_{\Phi}$};
		
	\node
	[
		draw,
		ellipse,
		minimum height=3.5cm,
		minimum width=2.5cm,
		left=of Zero.{east},
		anchor=east
	]
		(Permitted)
		{};
		\node[anchor=north west]  at (Permitted.135) {$\PermittedEvents_{\LTAlgebra_{\Phi}}$};
	\node
	[
		draw,
		dashed,
		ellipse,
		minimum height=4.1cm,
		minimum width=3.5cm,
		left=of Zero.{east},
		anchor=east
	]
		(EPermitted)
		{};
		\node[anchor=north west]  at (EPermitted.135) {$P$};
	\node
	[
		draw,
		ellipse,
		minimum height=3.5cm,
		minimum width=2.5cm,
		right=of Zero.{west},
		anchor=west
	]
		(Forbidden)
		{};
		\node[anchor=north east]  at (Forbidden.45) {$\ForbiddenEvents_{\LTAlgebra_{\Phi}}$};
	\node
	[
		draw,
		dashed,
		ellipse,
		minimum height=4.1cm,
		minimum width=3.5cm,
		right=of Zero.{west},
		anchor=west
	]
		(EForbidden)
		{};
		\node[anchor=north east]  at (EForbidden.45) {$F$};
	\node
	[
		draw,
		inner sep=10pt,
		fit={(Zero) (Permitted) (EPermitted) (Forbidden) (EForbidden)}
	]
		(LTAlgebra) {};
		\node[anchor=north west]  at (LTAlgebra.{north west}) {$\LTAlgebra_{\Phi}$};
\end{tikzpicture}
\end{center}
\caption{Algebraic Extension $(P, F)$ of $\Phi$ under $\Delta$}
\label{figure:algebraic:extension}
\end{figure}

Intuitively, on the algebraic side, ideals play the role that deductively closed sets of formulas play in  Reiter's notion of extension (c.f., \Cref{definition:extension:reiter}).
We bring attention to an important characteristic of the definition of an algebraic extension.
Algebraic extensions are ideals in a deontic action algebra.
This has the following implication.
In contrast to standard default reasoning where extensions are meta-level elements (deductively closed sets of formulas), algebraic extensions are semantic elements in the logic.

\begin{property}\label{proposition:algebraic:extension:existence}
	Algebraic extensions exist.
\end{property}

\begin{proof}
	The following notation is useful.
	If $(a, b)$ is a pair of elements, then $(a, b)_1 = a$ and $(a, b)_2 = b$.
	
	Define:
	\begin{nscenter}
	$
	\begin{array}{rl}
		e^0 &= (\PermittedEvents_{\LTAlgebra_{\Phi}}, \ForbiddenEvents_{\LTAlgebra_{\Phi}}) \\
		e^{(i+1)} &=
			\begin{cases}
				(\GeneratedIdeal[{e^i_1 \cup [\beta]_\Phi}], e^i_2)
					& \text{if there is $\Permitted[{\normal[\alpha][\beta]}] \in \Delta$ s.t.} \\
					& \text{\hspace{1em}
						${[\alpha]_\Phi \in e^i_1}$, 
						${\GeneratedIdeal[e^i_1 \cup {[\beta]_\Phi}] \cap e^i_2} = [0]_\Phi$, 
						$[\beta]_\Phi \not\preccurlyeq \DeonticDual{\PermittedEvents_{\LTAlgebra_{\Phi}}}$;}\\
				(e^i_1, \GeneratedIdeal[{e^i_2 \cup [\beta]_\Phi}])
					& \text{if there is $\Forbidden[{\normal[\alpha][\beta]}] \in \Delta$ s.t.} \\
					& \text{\hspace{1em}
						${[\alpha]_\Phi \in e^i_2}$,
						${\GeneratedIdeal[e^i_2 \cup {[\beta]_\Phi}] \cap e^i_1} = [0]_\Phi$,
						$[\beta]_\Phi \not\preccurlyeq \DeonticDual{\ForbiddenEvents_{\LTAlgebra_{\Phi}}}$;}\\
				e_i
					& \text{otherwise.}
			\end{cases}
	\end{array}
	$
\end{nscenter}
	
	Define:
	\begin{nscenter}
		$
		(P, F) =
			({\textstyle\bigcup\set{e^i_1}{i \geq 0}},
			{\textstyle\bigcup\set{e^i_2}{i \geq 0}}).
		$
	\end{nscenter}
	
	We claim that $(\PermittedEvents, \ForbiddenEvents)$ is an algebraic extension of $\Phi$ under $\Delta$.
	To prove this claim, first, we need to prove that $P$ and $F$ are ideals in $\LTAlgebra_{\Phi}$.
	This is direct.
	The proof continues by contradiction.
	Suppose that $(\PermittedEvents, \ForbiddenEvents)$ is not an algebraic extension of $\Phi$ under $\Delta$.
	Then, either (i), (ii), or (iii) from \Cref{definition:extension:algebraic} does not hold; or $(\PermittedEvents, \ForbiddenEvents)$ is not a fixed point of $\AlgebraicExtension_{\Phi}^{\Delta}$.
	The former cannot happen given the construction of the $e^i$'s. For the latter, we use two intermediate results:
	(a) the collection of ideals of a Boolean algebra form a complete lattice \cite{Halmos:2009}; and (b) $E^\delta_\Phi$ is monotone.
	This means that we can apply the Knaster-Tarski theorem (see e.g. \cite{davey:2002} for details).
	This yields a fixpoint:
		\begin{nscenter}
		$\textstyle \bigvee\limits_{\alpha < \omega_1} (E^\Delta_\Phi)^{\alpha}(\PermittedEvents_{\LTAlgebra_{\Phi}}, \ForbiddenEvents_{\LTAlgebra_{\Phi}}) = ({\textstyle\bigcup\set{e^i_1}{i \geq 0}},
					{\textstyle\bigcup\set{e^i_2}{i \geq 0}}).$
		\end{nscenter}
	From this, we obtain a contradiction
\end{proof}

\begin{property}\label{proposition:algebraic:extension:well-formed}
	If $(P, F)$ is an algebraic extension of $\Phi$ under $\Delta$, and $\Phi$ is {$\vdash$-consistent}, then, the triple $\langle \LTAlgebra_{\Phi}, P, F \rangle$ is a deontic action algebra.
\end{property}

\begin{proof}
	From \Cref{definition:extension:algebraic}, $P$ and $F$ are ideals.
	We only need to prove that ${P \cap F} = \{[0]_\Phi\}$.
	By contradiction, let
	${{P \cap F} = S} \neq \{[0]_\Phi\}$, consider the smallest ideals $P' \subseteq P$ and $F' \subseteq F$ s.t.\ ${P' \cap F'} = \{[0]_\Phi\}$.
	If we apply the function $\AlgebraicExtension_{\Phi}^{\Delta}$ to $(P',F')$ we obtain either a tuple $(P'', F')$ with $P'' \subsetneq P'$, or a tuple $(P', F'')$ with $F'' \subsetneq F'$.
	Note that, from our suppositions, we cannot have $\AlgebraicExtension_{\Phi}^{\Delta}(P',F') = (P',F')$.
	In the first case, we must have ${P'' \cap F'} = \{[0]_\Phi\}$, and similarly for the other case.
	Thus, $(P',F')$ is not the smallest subset of $(P,F)$ satisfying ${P\cap F} = \{[0]_\Phi\}$.
	This yields a contradiction.
\end{proof}

We define the notion of algebraic deontic default consequence in \Cref{definition:default:consequence:algebraic}.

\begin{definition}\label{definition:default:consequence:algebraic}
Let $\Phi$ be a $\entails$-consistent set of formulas and $\Delta$ be a set of basic deontic defaults; we say that a formula $\varphi$ of $\Formulas$ is an \emph{algebraic deontic default consequence} of $\Phi$ under $\Delta$, and write $\Phi \dconsequence_{\Delta} \varphi$, iff there exists an algebraic extension $(P, F)$ of $\Phi$ under $\Delta$ s.t.\ for all $(P', F') \supseteq (P,F)$, if ${P' \cap F'} = \{[0]_\Phi\}$, then
	$
	{\langle \LTAlgebra_{\Phi}, P', F' \rangle, \InterpretationFunction_{\Phi}} \consequence \varphi
	$
\end{definition}

We are now ready to show the main result of this work: the proof that default provability as defined in \Cref{definition:default:provability} is complete w.r.t.\ algebraic default consequence as defined in \Cref{definition:default:consequence:algebraic}.

\begin{theorem}\label{proposition:completeness:default}
	Let $\Phi$ be a $\entails$-consistent set of formulas and $\Delta$ be a set of basic deontic defaults; it follows that
	if $\Phi \dconsequence_{\Delta} \varphi$, then, $\Phi \dentails_{\Delta} \varphi$.
\end{theorem}

\begin{proof}
	We prove the contrapositive, i.e., if $\Phi \not \dentails_{\Delta} \varphi$, then $\Phi \not \dconsequence_{\Delta} \varphi$.
	Let $\Phi \not \dentails_{\Delta} \varphi$; the proof is concluded if $\Phi \not \dconsequence_{\Delta} \varphi$. This requires us to prove that for every algebraic extension $(P,F)$, there is $(P',F') \supseteq (P,F)$ s.t.\ ${P' \cap F'} = \{[0]_\Phi\}$ and $\langle \LTAlgebra_\Phi, P', F'\rangle \consequence \neg \varphi$.
	The proof continues by induction on $\varphi$.
	We assume that $\varphi$ is in \emph{negation normal form}.%
	\footnote{Any formula in \textsf{DAL} is equivalent to a formula in negation normal form.}
	Let $\varphi = \Permitted[\alpha]$; we must have $\langle \LTAlgebra_\Phi, P, F \rangle \consequence \Not \Permitted[\alpha]$ otherwise we would have either $\Phi \entails \Permitted[\alpha]$, or obtain $\Permitted[\alpha]$ by default detachment, but either case contradicts the assumption that $\Phi \not\dentails_\Delta \varphi$.
	For $\varphi = \Forbidden[\alpha]$ the proof is similar.
	Now, consider the case $\varphi = \neg \Permitted[\alpha]$;
	for any algebraic extension $(P, F)$ of $\Phi$ under $\Delta$, we define $P' = \GeneratedIdeal[{P \cup [\alpha]_\Phi}]$ and $F' = F$.
	Notice that if $\alpha=0$, we can trivially conclude ${P' \cap F} = \{[0]_\Phi\}$. On the other hand, if $\alpha\neq 0$, we cannot have $\Phi \dentails_\Delta \Forbidden[\alpha]$; otherwise we could obtain $\Phi \dentails_\Delta \neg \Permitted[\alpha]$, which contradicts the assumption that $\Phi$ is $\entails$-consistent.
	This also means that no default can add $\Forbidden[\alpha]$ to $\Phi$; thus $[\alpha]_\Phi \notin F$.
	This enables us to conclude that ${P' \cap F} = \{[0]_\Phi\}$.
	Then, $\langle \LTAlgebra_{\Phi}, P', F \rangle$ is a deontic action algebra s.t.\ $\langle \LTAlgebra_{\Phi}, P', F \rangle \consequence \Permitted[\alpha]$.
	The case $\varphi = \Not \Forbidden[\alpha]$ is similar.
	The result follows by a direct application of the inductive hypothesis to the cases $\varphi = \varphi' \vee \varphi''$ and $\varphi = \varphi' \wedge \varphi''$.
\end{proof}



\section{Illustrating Example}\label{sec:examples}
We illustrate the application of the formal framework introduced in earlier sections via a simple example.

We start by defining the vocabulary of basic action symbols as the set:
	\begin{nscenter}	
		$
		\BasicActions =
		\{
			\mathsf{d},
			\mathsf{o}
		\}.
		$
	\end{nscenter}
We use $\mathsf{d}$ to represent the action of driving on the road, and $\mathsf{o}$ to represent the action of overtaking (passing another car driving in the same direction).
We add as a basic principle on actions the following formula:
	\begin{nscenter}
		$
		(\mathsf{d} \equiv \mathsf{o}) = 0.
		$
	\end{nscenter}
Intuitively, this formula states that it is impossible to be simultaneously driving on the road and overtaking another car, or to be simultaneously not driving on the road and not overtaking another car.
This is all we know about actions.

Consider now a scenario in which we have the following regulations: it is permitted to drive on the road, which we formalize as $\Permitted[\mathsf{d}]$, and it is permitted by default to overtake a car whenever it is permitted to drive on the road, which we formalize as $\Permitted[{\normal[\mathsf{d}][\mathsf{o}]}]$.
We formally reason about this scenario as follows.
Let
	\begin{nscenter}
	$
	\begin{array}{ll}
		\Phi &=
		\{(\mathsf{d} \equiv \mathsf{o}) = 0, \Permitted[\mathsf{d}]\}
		\\
		\Delta &=
		\{ \Permitted[{\normal[\mathsf{d}][\mathsf{o}]}] \}.
	\end{array}
	$
	\end{nscenter}
The cube in \Cref{figure:algebraic:extension:example} depicts the Lindebaum-Tarski algebra $\LTAlgebra_{\Phi}$ of $\Phi$.
In this cube, nodes are labelled by equivalence classes under $\equiv_{\Phi}$.
The left face of the cube, highlighted in light gray, indicates the permitted ideal $\PermittedEvents_{\LTAlgebra_{\Phi}}$ of $\LTAlgebra_{\Phi}$, i.e., the set
\begin{nscenter}
	$\PermittedEvents_{\LTAlgebra_{\Phi}} = \GeneratedIdeal[{[\mathsf{d}]}_{\Phi}]$.
\end{nscenter}
The forbidden ideal of $\LTAlgebra_{\Phi}$ is the set $\ForbiddenEvents_{\LTAlgebra_{\Phi}} = \{[0]_\Phi\}$.
The deontic duals of $\PermittedEvents_{\LTAlgebra_{\Phi}}$ and $\ForbiddenEvents_{\LTAlgebra_{\Phi}}$ are $\DeonticDual{\PermittedEvents_{\LTAlgebra_{\Phi}}} = \DeonticDual{\ForbiddenEvents_{\LTAlgebra_{\Phi}}} = \emptyset$.
The pair
	\begin{nscenter}
		$(P, F) = (\LTAlgebra_{\Phi}, \ForbiddenEvents_{\LTAlgebra_{\Phi}})$
	\end{nscenter}
is the sole algebraic extension of $\Phi$ under $\Delta$.
From the above, it is possible to prove that:
	\begin{nscenter}
	$ \text{(i)~} \Phi
		\dconsequence_{\Delta}
			{\Permitted[\mathsf{d} \sqcup \mathsf{o}]}$.
	\end{nscenter}

\begin{figure}[!h]
\begin{center}

\def\aA{\mathsf{d}} 
\def\aB{\mathsf{o}} 

\def\Depth{5.5}
\def\Height{4}
\def\Width{4}

\begin{tikzpicture}
\coordinate (O) at (0,0,0);
\coordinate (A) at (0,\Width,0);
\coordinate (B) at (0,\Width,\Height);
\coordinate (C) at (0,0,\Height);
\coordinate (D) at (\Depth,0,0);
\coordinate (E) at (\Depth,\Width,0);
\coordinate (F) at (\Depth,\Width,\Height);
\coordinate (G) at (\Depth,0,\Height);

\draw[] (O) -- (C) -- (G) -- (D) -- cycle;
\draw[] (O) -- (A) -- (E) -- (D) -- cycle;
\draw[fill=gray!30,opacity=.8] (O) -- (A) -- (B) -- (C) -- cycle;
\draw[] (D) -- (E) -- (F) -- (G) -- cycle;
\draw[] (A) -- (B) -- (F) -- (E) -- cycle;


\node[anchor={north west}] at (O) {$[\aA \sqcap \overline{\aB}]_\Phi$};
\node[anchor={north west}] at (A) {$[\aA]_\Phi$};
\node[anchor={north west}] at (B) {$[\aA \sqcap \aB]_\Phi$};
\node[anchor={north west}] at (C) {$[0]_\Phi$};
\node[anchor={north west}] at (D) {$[\aA \not\equiv\aB]_\Phi$};
\node[anchor={north west}] at (E) {$[1]_\Phi$};
\node[anchor={north west}] at (F) {$[\aB]_\Phi$};
\node[anchor={north west}] at (G) {$[\overline{\aA} \sqcap \aB]_\Phi$};
\end{tikzpicture}
\end{center}
\caption{Algebraic Extension $(P,F)$ of $\Phi$ under $\Delta$}
\label{figure:algebraic:extension:example}
\end{figure}
	
Suppose that, to the scenario above, we add the fact that it is not permitted to overtake, formalized as $\Not \Permitted[\mathsf{o}]$, e.g., because the road is under construction.
Let $\Phi' = {\Phi \cup \{\Not \Permitted[\mathsf{o}]\}}$; we have that $\LTAlgebra_{\Phi'} = \LTAlgebra_{\Phi}$, i.e., the Lindenbaum-Tarski algebra of $\Phi'$ and $\Phi$ coincide.
Turning to permitted and forbidden ideals, we have ${\PermittedEvents_{\LTAlgebra_{\Phi'}} =\PermittedEvents_{\LTAlgebra_{\Phi}}}$ and ${\ForbiddenEvents_{\LTAlgebra_{\Phi'}} = \ForbiddenEvents_{\LTAlgebra_{\Phi}}}$.
As to deontic duals, we have $\DeonticDual{\PermittedEvents_{\LTAlgebra_{\Phi}}} = \{[\mathsf{o}]_\Phi, [\mathsf{\overline{r}} \sqcap \mathsf{o}]_\Phi\}$ and $\DeonticDual{\ForbiddenEvents_{\LTAlgebra_{\Phi}}} = \emptyset$.
The pair
	\begin{nscenter}
		$(P, F) = (\PermittedEvents_{\LTAlgebra_{\Phi'}},\ForbiddenEvents_{\LTAlgebra_{\Phi'}})$
	\end{nscenter}
is the sole algebraic extension of $\Phi'$ under $\Delta$.
From this it is possible to prove that:
	\begin{nscenter}
	$\text{(ii)~} \Phi \cup {\Not \Permitted[\mathsf{o}]} 
			\not\dconsequence_{\Delta}
				{\Permitted[\mathsf{d} \sqcup \mathsf{o}]}$.
	\end{nscenter}

When taken together, (i) and (ii) illustrate some of the ``dynamic'' behaviour of our framework for default reasoning over deontic action operations.
If all we know is that driving on the road is permitted, and we have no information on whether overtaking is not permitted, we can conclude by default that the free choice of driving on the road or overtaking is permitted.
However, this conclusion is withdrawn as soon as we learn that overtaking is not permitted.

On the syntactical side, let $\rmEpsilon = \Phi \cup \{\Permitted[\mathsf{o}]\}$, it is possible to prove that the $\ClosureDAL{\rmEpsilon}$ is the sole extension of $\Phi$ under $\Delta$. From this fact, we can conclude that $\Phi \econsequence_\Delta \Permitted[\mathsf{d} \sqcup \mathsf{o}]$ using the sequence $s$ below as a $\entails$-proof witness.
	\begin{nscenter}
		\begin{tabular}{r l r}
			1. & $\Permitted[\mathsf{o}]$ & from $\rmEpsilon$; \\
			2. & $\Permitted[\mathsf{d}]$ & from $\rmEpsilon$; \\
			3. & $\Permitted[\mathsf{d}] \And \Permitted[\mathsf{o}]$ & from 1.\ and 2.\ in \CPL; \\
			4. & $(\Permitted[\mathsf{d}] \And \Permitted[\mathsf{o}]) \Implies \Permitted[\mathsf{d} \sqcup \mathsf{o}]$ & from D1.\ in \DAL; \\
			5. & $\Permitted[\mathsf{d} \sqcup \mathsf{o}]$ & \emph{modus ponens} on 3.\ and 4. \\
		\end{tabular}
	\end{nscenter}
The sequence $s$ can be transformed into a default proof of $\Permitted[\mathsf{d} \sqcup \mathsf{o}]$ from $\Phi$ under $\Delta$ by appending $\Permitted[\mathsf{d}]$, i.e., a proof of the prerequisite of the default in $\Delta$, at the beginning of $s$.
The resulting default proof, shown below, illustrates the construction used in \Cref{proposition:default:proof:extensions}.
	\begin{nscenter}
		\begin{tabular}{r l r}
			1. & $\Permitted[\mathsf{d}]$
				& from $\Phi$ \\
			2. & $\Permitted[\mathsf{o}]$
				& \emph{default detachment} on 1.\ and $\Permitted[{\normal[\mathsf{d}][\mathsf{o}]}]$ \\
			3. & $\Permitted[\mathsf{d}]$
				& from $\Phi$ \\
			4. & $\Permitted[\mathsf{d}] \And \Permitted[\mathsf{o}]$
				& from 2.\ and 3.\ in \CPL \\
			5. & $(\Permitted[\mathsf{d}] \And \Permitted[\mathsf{o}]) \Implies \Permitted[\mathsf{d} \sqcup \mathsf{o}]$
				& from D1.\ in \DAL \\
			6. & $\Permitted[\mathsf{d} \sqcup \mathsf{o}]$
				& \emph{modus ponens} on 4.\ and 4. \\
		\end{tabular}
	\end{nscenter}


\section{Final remarks}
\label{sec:final}

We introduced a novel presentation of a Default Logic on Segerberg's Deontic Action Logic \DAL \cite{Segerberg1982}.
This formalism enables us to reason about scenarios involving norms defined on actions, and eventual changes in such norms.
In addition to a standard construction of building a default logic over an underlying logic, as done, e.g., in \cite{Cassano:2019}, our approach uses the semantic elements of \DAL, i.e., deontic Boolean algebras, in a natural way to capture the meaning of defaults.
The first benefit of our approach is its simplicity.
Default logic heavily makes use of fix-point constructions.
Viewing default logic from an algebraic perspective allow us to rely on well-known results, i.e., the Knaster-Tarski theorem, to prove the existence of fix-points.
The second, and most important, benefit of our approach is that it allows us to extend the ideas presented in \cite{Segerberg1982} to obtain a completeness theorem for normal default consequence on basic deontic defaults.

An important reference in the area of default reasoning and deontic logic is \cite{Nute:1997}. Therein, several authors present diverse approaches to defeasible reasoning over normative systems.
Interestingly, some papers of these approaches investigate the combination of some of Reiter's notions with deontic logic.
For instance, in \cite{Horty:1997}, Horty uses a non-normal modal logic, based on Chellas' ideas \cite{Chellas:1980}, and combines this formalism with Reiter's default logic to adapt the notion of obligation to non-monotonic reasoning.
Another example is \cite{vanderTorre:1980}. In this work, the authors use Reiter's approach to discuss three kinds of defeasibility: \emph{factual defeasibility}, \emph{overriding defeasibility}, and \emph{weak-overriden defeasibility}.
A defeasible deontic logic based also on Reiter's notion of an extension is presented in \cite{Ryu:1980}. In this case, the formalism tackles a notion of preference between norms to deal conflictive rules caused by many different sources.
Another non-monotonic logic to formalize and reason about prima-facie obligations is presented in \cite{Asher:1980}.
Finally, in \cite{Royakkers:1980} there is a proposal to distinguish between default rules and norms, also using Reiter's notions.
It is important to remark that all these works are focused on Standard Deontic Logic, i.e., they are based on an ought-to-be deontic logic (i.e., deontic operators applied to propositions).
In this paper, we incorporate default reasoning to an ought-to-do deontic logic.
As discussed in \cite{Castaneda:1972},  ought-to-do deontic logics are orthogonal to ought-to-be formalisms; in the former, the prescriptions are applied to actions; while in the latter, norms are applied to ``state of affairs''.
The interested reader is referred to the aforementioned work for an in-depth discussion about the implications of this difference.
We are not aware of any work that provides default reasoning over an ought-to-do deontic logic.
It must also be noticed that our approach has a semantic flavour, in contrast to the works mentioned above which make use of Reiter's notion of extensions, which is syntactical in nature.

A more recent reference concerning non-monotonic reasoning on logics handling actions is \cite{castilho:2002}.
In this work, the authors present what they call a \emph{Logic of Actions and Plans with Dependences} (\LAPD).
In \LAPD, actions influence the truth of certain propositions and dependences capture certain frame conditions on the execution of actions.
Intuitively, this logic may be seen as capturing the effects of executing an action in certain contexts; similarly to our deontic defaults.
However, from a purely technical perspective, \LAPD is more related to Meyer's approach to deontic operators applied to actions rather than to Segerberg's approach.
In any case, we would like to establish more concrete connections between logics such as \LAPD and our logic.
In particular, it would be interesting to study the difference in expressivity and computational complexity.
This is part of the further work that we plan to undertake.

There are several other interesting directions to explore in the future.
From the theoretical side, it would be interesting to obtain a completeness result for defaults whose prerequisites, justifications, and consequents can be any formula of \DAL.
We conjecture that for normal defaults of the form $\normal[{\Permitted[\alpha]}][{\Forbidden[\beta]}]$ or $\normal[{\Forbidden[\alpha]}][{\Permitted[\beta]}]$, our results can be easily extended.
The additional complication originates from the use of equalities $=$ and negation $\Not$ in the consequent of an arbitrary default; e.g., $=$ may change the Lindenbaum-Tarski algebra used in the definition of an algebraic extension.
Also, it would be interesting to extend the basic logic described here with additional deontic operators.
Immediately coming to mind are those of obligation, weak permission or prohibition, and conditional prescriptions.
These new deontic operators can also be dealt algebraically by means of generalizations of Boolean algebras.
Two simple examples are Boolean algebras with operators, used to algebraize modalities \cite{Blackburn:2007}, and residuated Boolean algebras \cite{Jipsen:1992}, which provide residual operators useful for reasoning about action composition.

Turning to practical considerations, it would be interesting provide some tool support for reasoning in our setting.
The algebraic semantics of \DAL, as well as the extensions mentioned above, can be modelled using First-Order theories and are therefore amenable to the use of SMT solvers \cite{Barret:2009}.
We would like to have at hand implementations for reasoning about conditions on ideals and the existence of algebraic extensions.
The symbolic representation of Boolean formulas such as binary decision diagrams \cite{Bryant:2018} can be used to provide efficient ways of encoding deontic formulas and defaults.
However, these are just some preliminary thoughts that require further exploration.

\paragraph*{Ackowledgements.} This work was partially supported by ANPCyT-PICTs-2017-1130 and 2016-0215, MinCyT C\'ordoba, SeCyT-UNC, the European Union’s Horizon 2020 research and innovation programme
under the Marie Skodowska-Curie grant agreement No. 690974 for the project MIREL:
MIning and REasoning with Legal texts, and the Laboratoire International Associ\'e INFINIS.

\bibliographystyle{eptcs}
\bibliography{bibliography}
\end{document}